\documentclass[12pt]{article}
\usepackage{epsfig, graphicx} 
\usepackage{amssymb, amsthm}
\usepackage{amsmath}
\usepackage{relsize}

\usepackage{calrsfs}
\DeclareMathAlphabet{\pazocal}{OMS}{zplm}{m}{n}

\newcommand{\lln}{\lambda_+}
\newcommand{\lrn}{\lambda_-}

\usepackage{calrsfs}  

\usepackage{booktabs}
\usepackage{enumitem} 	
\usepackage{eurosym}	
\usepackage{wrapfig}	
\usepackage{float}		
\usepackage{subfig}		
\usepackage{bbm}		
\usepackage{footnote}
\usepackage{thmtools}
\usepackage[title]{appendix}
\usepackage{xcolor}
\usepackage{url}

\usepackage{hyperref}
\usepackage[authoryear]{natbib} 
\bibpunct{(}{)}{,}{a}{}{,}

\newtheorem{theorem}{Theorem}

\newtheorem{lemma}[theorem]{Lemma}

\theoremstyle{definition}
\newtheorem{definition}[theorem]{Definition}
\theoremstyle{remark}

\numberwithin{theorem}{section}
\numberwithin{proposition}{section}
\numberwithin{lemma}{section}
\numberwithin{corollary}{section}
\numberwithin{definition}{section}
\numberwithin{remark}{section}
\numberwithin{example}{section}

\newcommand{\be}{\begin{equation}}
\newcommand{\en}{\end{equation}}
\newcommand{\ben}{\begin{equation*}}
\newcommand{\enn}{\end{equation*}}
\newcommand{\bea}{\begin{eqnarray}}
\newcommand{\ena}{\end{eqnarray}}

\begin{document}
	\newlength\tindent
	\setlength{\tindent}{\parindent}
	\setlength{\parindent}{0pt}
	\renewcommand{\indent}{\hspace*{\tindent}}
	
	\begin{savenotes}
		\title{
			\bf{ Smile asymptotics for \\
            Bachelier implied volatility
		}}
		\author{
			Roberto Baviera$^\ddagger$  \&
			Michele Domenico Massaria$^\ddagger$ 
		}
		
		\maketitle
		
		\begin{tabular}{ll}
			$(\ddagger)$ &  Politecnico di Milano, Department of Mathematics, 32 p.zza L. da Vinci, Milano \\
		\end{tabular}
	\end{savenotes}
	
	\vspace{0.5cm}
	
	\begin{abstract}	
        \noindent
We investigate the asymptotic behaviour of the Bachelier implied volatility {\color{black} tails}, extending the large-strike results established for the Black-Scholes implied volatility. Exploiting the theory of regular variation, we derive explicit expressions for the Bachelier implied volatility in the wings of the smile, directly linking them to the tail decay of the underlying returns' distribution.
Furthermore, we establish a rigorous connection between the analyticity strip of the characteristic function and the asymptotic slope of the volatility smile. Moreover, we show that if the implied variance grows linearly for large absolute moneyness degree, the underlying returns' distribution must exhibit exponential tail decay, and the corresponding characteristic function is analytic in a horizontal strip of the complex plane. These findings \textcolor{black}{characterise} valid models based solely on the observable asymptotic behaviour of the smile.
         \end{abstract}
	
	\vspace*{0.11truein}
	\vspace*{0.11truein}
	
	
	\vspace{1.5cm}
	\begin{flushleft}
		{\bf Address for correspondence:}\\
		Roberto Baviera\\
		Department of Mathematics \\
		Politecnico di Milano\\
		32 p.zza Leonardo da Vinci \\ 
		I-20133 Milano, Italy \\
		Tel. +39-02-2399 4575\\
		roberto.baviera@polimi.it
	\end{flushleft}

\newpage
\begin{center}
\Large\bfseries 
			\bf{ Smile asymptotics for Bachelier implied volatility
		}
\end{center}
\section{Introduction}
\textcolor{black}{In recent years, the Bachelier model \citep{bachelier1900jeu,davis2006} has regained relevance in financial modelling, particularly in markets where the underlying can take negative values, such as interest rates and certain commodity futures \citep{CMEGroup,ICE}. \textcolor{black}{While exponential L\'evy models remain the natural choice when the underlying is constrained to stay positive, as in equity or FX markets, an arithmetic specification of the underlying is required whenever negative values can occur.} Given this renewed interest, a refined understanding of the Bachelier implied volatility has therefore become increasingly relevant.
}

\smallskip

\textcolor{black}{In the context of implied volatility asymptotics, \cite{lee2004moment} established a fully model-free upper bound for the asymptotic slope (in the form of a $\limsup$) of the Black-Scholes implied variance; a similar result for the Bachelier implied volatility has been obtained by \cite{floc2022bachelier}. The work of \cite{benaim2009regular} has significantly advanced the theoretical understanding of the smile behaviour for the Black-Scholes implied volatility. By leveraging the theory of regular variation \citep{bingham1987regular}, they established a rigorous connection between the tail behaviour of the underlying's distribution and the wings of the implied volatility surface, providing a powerful extension of Lee's moment formula. In particular, under suitable regularity conditions on the tails of the distribution, satisfied by most models of practical interest, the $\limsup$ in Lee's result can be replaced by an actual limit, yielding a precise characterisation of the asymptotic slope. A detailed description of this problem can be found in excellent textbooks \citep[see, e.g.,][Ch.7]{gatheral2011volatility}.
}

\smallskip

\textcolor{black}{This paper makes three main contributions.}

\smallskip

The first contribution of this paper is to extend the result of \cite{benaim2009regular} to the Bachelier implied volatility. In particular, we investigate the asymptotic behaviour of the Bachelier implied volatility for large moneyness degree. We show that, under appropriate assumptions on the distribution of the underlying's returns, the wings of the Bachelier implied volatility can be characterised through explicit formulas.

\smallskip

The second contribution of this work is to establish a direct connection between the asymptotic behaviour of the implied volatility smile and the analyticity properties of the characteristic function of returns. In particular, following \cite{benaim2008smile}, that related the behaviour of the tails of a probability distribution to the corresponding moment generating function, we show that the boundaries of the analyticity strip in the complex plane determine the exponential decay rate of returns' distribution tails, which are related to the behaviour of Bachelier implied volatility for large absolute values of moneyness degree.

\smallskip

The third contribution is to show that the asymptotic behaviour of the implied volatility determines the behaviour of the tails of the returns' probability distribution. In particular, we prove that if the implied variance grows linearly in the absolute value of the moneyness degree, for large moneyness degree, then the returns' probability distribution tails exhibit exponential decay, and the corresponding characteristic function is analytic in a horizontal strip in the complex plane. \textcolor{black}{To illustrate the scope of these results, we show that the assumptions they rely on, and therefore the results themselves, are satisfied by a broad range of models of practical interest, e.g., the Meixner, Tempered Stable and Generalized Hyperbolic processes (that include the well-known Normal Inverse Gaussian and Variance Gamma models) and their additive extensions. Moreover, we illustrate the accuracy of these results through a numerical application to the Normal Inverse Gaussian and Meixner models.}

\smallskip

The rest of the paper is organised as follows. 
In Section \ref{sec:Bachelier_regVar}, we recall the Bachelier formula for European options and the theory of regular varying functions.
In Section \ref{sec:smiles}, we obtain the asymptotic behaviour of the Bachelier implied volatility for large strikes, and we show that if the implied variance is asymptotic to the moneyness degree, then the returns' characteristic function is analytic in a horizontal strip of the complex plane, and we illustrate these results through some model examples and a numerical application. Finally, we state our conclusions in Section \ref{sec:Conclusions}.

\section{Bachelier formula and regular varying functions}
\label{sec:Bachelier_regVar}
\textcolor{black}{Let $(\Omega,\mathcal{F},\mathbb{F},\mathbb{Q})$ be a filtered probability space, with $\mathbb{F}=(\mathcal{F}_t)_{t\geq0}$ a complete and right-continuous filtration, where $\mathbb{Q}$ denotes the risk-neutral measure. We let $\{F_t\}_{t\geq0}$ denote an $\mathbb{R}$-valued, $\mathbb{F}$-adapted process representing the evolution of the underlying.\\
Throughout, we consider European options with maturity $t>0$. To allow for negative values of the underlying, we work within an arithmetic framework and specify the terminal value under the risk-neutral measure as
\begin{equation}
\label{eq:fwd_dynamics}
    F_t \overset{\mathrm{law}}{=} F_0+\sqrt{t}\,\zeta,
\end{equation}
where $\zeta$ is an $\mathbb{R}$-valued random variable representing the risk-neutral return.
Given a strike $K$, we define the moneyness degree $\kappa$
\[
\kappa:=\frac{K-F_0}{\sqrt{t}}\,\,.
\]
}
If we indicate with 
$F(\bullet)$ the cdf of $\zeta$
and $\bar{F}(\bullet):=1-F(\bullet)$, then its undiscounted call price can be written as
\begin{equation}
\label{eq:call_integral}
    c(\kappa):=\mathbb{E}\left[\zeta-\kappa\right]^+=\int_\kappa^{+\infty}(\zeta-\kappa)\mathrm{d}F(\zeta)\,\,.
\end{equation}

The normalised Bachelier call price $c_b$ corresponds to the case with $\zeta$ distributed as a zero mean Gaussian rv with variance $\sigma^2$ \citep[see, e.g.,][eq.(3)]{choi2022black}
\begin{equation}
    \label{eq:norm_Bachelier}
    c_b(\kappa,\sigma):=-\kappa\,\Phi\left(-\frac{\kappa}{\sigma}\right)+\sigma\,\varphi\left(-\frac{\kappa}{\sigma}\right)\,\,,
\end{equation}
with $\varphi(\bullet)$ and $\Phi(\bullet)$, respectively, the pdf and the cdf of a standard normal rv.\\
The implied volatility is the (unique) value \footnote{The existence of a unique implied volatility for the Bachelier formula is a consequence of the continuity and monotonicity in $\sigma>0$ of the function $\sigma\mapsto c_b(\kappa,\sigma):\mathbb{R}^+\setminus\{0\}\to\mathbb{R}^+\setminus\{0\}$.} $\mathrm{I}(\kappa)$ such that
\begin{equation}
\label{eq:IV_eq}
c_b\left(\kappa,\mathrm{I}(\kappa)\right)=c(\kappa)\,\,.
\end{equation}
This setting is an extension to real-valued underlying of \cite{benaim2009regular}, who consider the relation for the Black-Scholes formula.\\ 
Similarly to \eqref{eq:norm_Bachelier}, the normalised Bachelier put price is
\begin{equation}
    p_b(\kappa,\sigma):=\kappa\left(1-\,\Phi\left(-\frac{\kappa}{\sigma}\right)\right)+\sigma\varphi\left(-\frac{\kappa}{\sigma}\right)\,\,,
\end{equation}
that can be written as in \eqref{eq:call_integral} and \eqref{eq:IV_eq} as follows
\begin{equation}
\label{eq:put_integral_cdf}
    p_b\left(\kappa,\mathrm{I}(\kappa)\right)=p(\kappa):=\int_{-\infty}^\kappa(\kappa-\zeta)\mathrm{d}F(\zeta)\,\,.
\end{equation}
The following Lemma identifies upper and lower bounds for the normalised Bachelier call formula.
\begin{lemma} {\rm (Bachelier formula bounds)} 
\label{lem:BoundsBachelier}\\  
For any $\beta>0$, $y>0$
\[
e^{-\frac{1}{2 \beta} y} g_l(y) \le c_b\left(y, \sqrt{\beta y}\right) \le e^{-\frac{1}{2 \beta} y} g_u(y)\,\,,
\]
where $g_l(y)$ and $g_u(y)$ are positive and go to $+\infty$ as $y\to+\infty$.
\end{lemma}
\begin{proof}
    Utilising the (normalised) Bachelier formula \eqref{eq:norm_Bachelier}, we have
    \begin{equation*}
        c_b\left(y, \sqrt{\beta y}\right)=\displaystyle-y\,\Phi\left(-\sqrt{\frac{y}{\beta}}\right)+\sqrt{\beta\,y}\,\varphi\left(-\sqrt{\frac{y}{\beta}}\right)\leq\sqrt{\frac{\beta\,y}{2\pi}}e^{-\frac{1}{2\beta}y}\,\,,
    \end{equation*}
    where the inequality holds because the first addend is negative.\\
    This proves the upper bound choosing
    \[
    g_u(y):=\sqrt{\frac{\beta\,y}{2\pi}}\,\,.
    \]
    For the lower bound, we recall that, for $x>0$, the following inequality holds \citep[see, e.g.,][Sec.2.26, p.179]{mitrinovic1970analytic}
    \begin{equation*}
        \Phi(-x)\leq\frac{1}{x+\sqrt{x^2+\frac{4}{\pi}}}\varphi(-x)\,\,,
    \end{equation*}
    thus,
    \begin{equation}
        c_b\left(y, \sqrt{\beta y}\right)\geq e^{-\frac{1}{2\beta}y}\frac{1}{\sqrt{2\pi}}\left\{-\frac{y}{\sqrt{\frac{y}{\beta}}+\sqrt{\frac{y}{\beta}+\frac{4}{\pi}}}+\sqrt{\beta\,y}\right\}=g_l(y)e^{-\frac{1}{2\beta}y}\,\,,
    \end{equation}
    with
    \begin{equation*}
            g_l(y):=\frac{\sqrt{\beta y}\sqrt{\frac{y}{\beta}+\frac{4}{\pi}}}{\sqrt{\frac{y}{\beta}}+\sqrt{\frac{y}{\beta}+\frac{4}{\pi}}}
    \end{equation*}
\end{proof}
We recall the notion of regularly varying functions. We provide the formal definition \citep[cf.,][p.18]{bingham1987regular} and state a result from \cite{bingham1987regular} concerning their asymptotic behaviour.\\
\begin{definition}
    A positive real-valued measurable function $g$ is regular varying with index $\theta\geq0$, in symbols $g\in\mathcal{R}_\theta$, if $\forall\,\mu>0$
    \begin{equation*}
        \lim_{x\to+\infty}\frac{g(\mu x)}{g(x)}=\mu^\theta\,\,.
    \end{equation*}
\end{definition}

\bigskip

It holds that if $g\in\mathcal{R}_\theta$ and $g\sim h$ for $x\to+\infty$,\footnote{$g\sim h$ for $x\to+\infty$ means $\lim_{x\to+\infty}\frac{g(x)}{h(x)}=1$.} then $h\in\mathcal{R}_\theta$.\\
The following result can be found in \citet[Thm.4.12.10, p.255]{bingham1987regular}.
\begin{theorem}\label{thm:Bingham}
    Let $g\in\mathcal{R}_\theta$. Then, for $x\to+\infty$
    \begin{equation*}
        -\ln\int_x^{+\infty}e^{-g(y)}\mathrm{d}y\sim g(x)\,\,.
    \end{equation*}
\end{theorem}
\section{Smile asymptotics}
\label{sec:smiles}

In this Section, we derive the main results of the paper concerning the asymptotic behaviour of the Bachelier implied volatility for large absolute values of moneyness degree. We establish a rigorous link between the decay rate of the underlying's return distribution tails and the slope of the implied variance for large absolute moneyness degree.

\smallskip

We proceed in two steps. First, we investigate the direct problem: assuming a specific tail behaviour for the returns (expressed in terms of regular variation), we deduce the corresponding asymptotic formula for the implied volatility in the wings of the smile. Subsequently, we address the inverse problem, showing that the asymptotic behaviour of the implied volatility itself determines the analytic properties of the characteristic function of returns. This latter result is particularly significant as it allows for model characterisation based on observable market data.

\smallskip

We begin by presenting a preliminary result that relates the call and put prices to the cdf of returns.

\begin{lemma}\label{lem:bcl_price_barF}
    The following holds for the call price $c(\kappa)$ in \eqref{eq:call_integral}
    \begin{equation*}
        c(\kappa)=\int_\kappa^{+\infty}\bar{F}(\zeta)\mathrm{d}\zeta\,\,,
    \end{equation*}
    while for the put price $p(\kappa)$ in \eqref{eq:put_integral_cdf}
    \begin{equation*}
        p(\kappa)=\int_{-\infty}^\kappa F(\zeta)\mathrm{d}\zeta\,\,.
    \end{equation*}
\end{lemma}
\begin{proof}
    Let us recall that 
    \[
    [\zeta-\kappa]^+=\int_{\kappa}^{+\infty}\mathbbm{1}_{\{\zeta>x\}}\mathrm{d}x\,\,.
    \]
    Since the integrand is non-negative, by Fubini's Theorem \citep[cf.][Thm.8.8, p.164]{rudin2006real}, we can interchange the order of expectation and integration:
    \[
    c(\kappa) = \int_\kappa^{+\infty} \mathbb{E}\left[ \mathbbm{1}_{\{\zeta > x\}} \right] \mathrm{d}x = \int_\kappa^{+\infty} \mathbb{P}(\zeta > x) \mathrm{d}x = \int_\kappa^{+\infty} \bar{F}(x) \mathrm{d}x\,\,,
    \]
    this proves the thesis. The proof for the put price is analogous
\end{proof}
We can introduce the following result on the asymptotic implied volatility smile, that is an extension to the Bachelier framework of Theorem 1.1 in \cite{benaim2009regular}. Before proving the Theorem, we demonstrate two technical Lemmas.

\begin{lemma}
\label{lem:tec_lem}{\rm (Right wing properties)}.\\
Assume that one of the following assumptions hold \footnote{\textcolor{black}{We recall the Landau notation:} $g(\kappa)=O(h(\kappa))$ means that, $\exists\,\kappa_0>0,\,C>0$ such that \textcolor{black}{$g(\kappa)\leq C\,h(\kappa)$} $\forall\,\kappa\geq\kappa_0$. $g(\kappa)=o(h(\kappa))$ as $\kappa\to+\infty$ means that $\lim_{\kappa\to+\infty}\frac{g(\kappa)}{h(\kappa)}=0$.}
\textcolor{black}{
\begin{align*}
    &\text{a)}\quad
\displaystyle\exists\,\epsilon>0\quad\text{such that}\quad c(\kappa)= O\left(e^{-\epsilon\kappa}\right)\quad\text{as}\quad\kappa\to+\infty\;\;,\\[3mm]
&\text{b)}\quad\displaystyle\exists\,\mu>0\quad\text{such that}\quad \mathrm{I}(\kappa)\sim\mu\sqrt{\kappa}\quad\text{as}\quad\kappa\to+\infty\;\;,
\end{align*}}
then the following asymptotic properties of the implied volatility $\mathrm{I}(\kappa)$ hold 
\begin{align*}
&\text{i)}\quad
\displaystyle\lim_{\kappa\to+\infty}\frac{\sqrt{\kappa^2+\frac{4}{\pi}\,\mathrm{I}(\kappa)^2}}{\kappa+\sqrt{\kappa^2+\frac{4}{\pi}\,\mathrm{I}(\kappa)^2}}=\lim_{\kappa\to+\infty}\frac{\sqrt{\kappa^2+2\,\mathrm{I}(\kappa)^2}}{\kappa+\sqrt{\kappa^2+2\,\mathrm{I}(\kappa)^2}}=\frac{1}{2}\;\;,\\[3mm]
&\text{ii)}\quad\displaystyle\ln\left(\mathrm{I}(\kappa)\right)+\ln\left(\frac{1}{2}\right)=o\left(\frac{\kappa^2}{2\,\mathrm{I}(\kappa)^2}\right)\quad\text{as}\;\;\kappa\to+\infty\;\;.
\end{align*}
\end{lemma}
\begin{proof}
    If a) holds, then, thanks to Lemma \ref{lem:BoundsBachelier}, $\forall\,\beta>0$ such that $\frac{1}{2\beta}\leq \epsilon$, \textcolor{black}{$\exists\,\kappa_0>0$, $C>0$
    \begin{equation}
    \label{eq:ratio_to_0}
        \frac{c_b(\kappa,\mathrm{I}(\kappa))}{c_b(\kappa,\sqrt{\beta\,\kappa})}\leq\frac{C\,e^{-\epsilon\kappa}}{g_l(\kappa)e^{-\frac{1}{2\beta}\kappa}}\leq\frac{C}{g_l(\kappa)}\rightarrow0^+\quad\text{as}\quad \kappa\rightarrow+\infty\,\,.
    \end{equation}}
    This means that, $\forall\,\epsilon>0$ such that a) holds, choosing $\beta=\frac{1}{2\epsilon}$ in \eqref{eq:ratio_to_0}, there exists $\kappa_0>0$ such that for all $\kappa > \kappa_0$, $\mathrm{I}(\kappa)<\sqrt{\frac{\kappa}{2\epsilon}}$,  thus
\begin{equation*}
\displaystyle\lim_{\kappa\to+\infty}\frac{\sqrt{\kappa^2+\frac{4}{\pi}\,\mathrm{I}(\kappa)^2}}{\kappa+\sqrt{\kappa^2+\frac{4}{\pi}\,\mathrm{I}(\kappa)^2}}=\displaystyle\lim_{\kappa\to+\infty}\frac{\sqrt{\kappa^2+2\,\mathrm{I}(\kappa)^2}}{\kappa+\sqrt{\kappa^2+2\,\mathrm{I}(\kappa)^2}}=\frac{1}{2}\;\;,
\end{equation*}
that proves i) if assumption a) holds.
\textcolor{black}{Moreover, in the proof of i) we have shown that, under assumption a), there exists $\kappa_0>0$ such that $I(\kappa)^2<\frac{\kappa}{2\epsilon}$ for all $\kappa>\kappa_0$, this bound gives that, for all $\kappa>\kappa_0$
\[
\frac{\kappa^2}{2\,I(\kappa)^2} > \epsilon\kappa \to +\infty \quad \text{as}\quad\kappa\to+\infty\,\,,
\]
let us define
\[
A(\kappa):=\frac{\kappa^2}{2\,I(\kappa)^2}\,\,.
\]
Since $\displaystyle I(\kappa)^2=\frac{\kappa^2}{2\,A(\kappa)}$, we have
\[
\ln(I(\kappa)) + \ln\left(\frac{1}{2}\right) = \ln\kappa - \frac{1}{2}\ln(A(\kappa)) - \frac{3}{2}\ln 2\,\,,
\]
and therefore
\[
\displaystyle\frac{\ln(I(\kappa))+\ln\left(\frac{1}{2}\right)}{A(\kappa)} = \frac{\ln\kappa}{A(\kappa)} - \frac{\ln(A(\kappa))}{2A(\kappa)} - \frac{3\ln 2}{2A(\kappa)} \to 0 \quad \text{as}\quad \kappa\to+\infty\,\,,
\]
since $\exists\,\kappa_0>0$ such that for all $\kappa>\kappa_0$, $A(\kappa)>\epsilon\kappa$ and $A(\kappa)\to+\infty$. Hence
\[
\ln I(\kappa)+\ln\left(\frac{1}{2}\right) = o\left(\frac{\kappa^2}{2\,I(\kappa)^2}\right)\quad\text{as}\quad\kappa\to+\infty\,\,,
\]
that proves ii) if assumption a) holds.}\\
\indent
Finally, under assumption b), properties i) and ii) follow straightforwardly from the asymptotic behaviour $\mathrm{I}(\kappa)\sim\mu\sqrt{\kappa}$ as $\kappa\to+\infty$, observing that $\mathrm{I}(\kappa)^2$ scales linearly with $\kappa$
\end{proof}
\begin{lemma}
\label{lem:tec_lem_left}{\rm (Left wing properties).}\\
Assume that one of the following holds
\textcolor{black}{
\begin{align*}
    &\text{a)}\quad
\displaystyle\exists\,\epsilon>0\quad\text{such that}\quad p(-\kappa)= O\left(e^{-\epsilon\kappa}\right)\quad\text{as}\quad\kappa\to+\infty\;\;,\\[3mm]
&\text{b)}\quad\displaystyle\exists\,\mu>0\quad\text{such that}\quad I(-\kappa)\sim\mu\sqrt{\kappa}\quad\text{as}\quad\kappa\to+\infty\;\;,
\end{align*}}
then the following asymptotic properties of the implied volatility $\mathrm{I}(\kappa)$ hold
\begin{align*}
&\text{i)}\quad
\displaystyle\lim_{\kappa\to+\infty}\frac{\sqrt{\kappa^2+\frac{4}{\pi}\,\mathrm{I}(-\kappa)^2}}{\kappa+\sqrt{\kappa^2+\frac{4}{\pi}\,\mathrm{I}(-\kappa)^2}}=\displaystyle\lim_{\kappa\to+\infty}\frac{\sqrt{\kappa^2+2\,\mathrm{I}(-\kappa)^2}}{\kappa+\sqrt{\kappa^2+2\,\mathrm{I}(-\kappa)^2}}=\frac{1}{2}\;\;,\\[3mm]
&\text{ii)}\quad\displaystyle\ln\left(\mathrm{I}(-\kappa)\right)+\ln\left(\frac{1}{2}\right)=o\left(\frac{\kappa^2}{2\,\mathrm{I}(-\kappa)^2}\right)\quad\text{as}\quad\kappa\to+\infty\;\;.
\end{align*}
\end{lemma}
\begin{proof}
The proof is analogous to Lemma \ref{lem:tec_lem}, considering that, $\forall\,\sigma>0$ and $\kappa\in\mathbb{R}$
\[
p_b(-\kappa,\sigma)=c_b(\kappa,\sigma)
\]
\end{proof}

The following Theorem extends to the Bachelier framework Theorem 1.1 in \cite{benaim2009regular}, which has been proven in the Black-Scholes framework.
\begin{theorem}\label{prop:right_wing}{\rm (Right wing formula)}\\
    Assume $\theta\geq1$. Then, for $\kappa\to+\infty$, the following implications hold 
    \begin{equation*}
\vspace{0.2cm}
        \text{i)}\;\;-\ln \bar{F}(\kappa)\in\mathcal{R}_\theta\;\;\Rightarrow\;\;\text{ii)}\;\;\begin{cases}
-\ln c(\kappa)\sim-\ln \bar{F}(\kappa)&\\
-\ln c(\kappa)\in\mathcal{R}_\theta&
\end{cases}
\hspace{-0.5cm}\;\;\Rightarrow\;\;\text{iii)}\;\;\frac{\mathrm{I}(\kappa)^2}{\kappa}\sim -\frac{\kappa}{2\ln \bar{F}(\kappa)}\,\,.
    \end{equation*}
\end{theorem}
\begin{proof}
    
    \smallskip

    We prove that i) $\Rightarrow$ ii). Thanks to Lemma \ref{lem:bcl_price_barF},
    \begin{equation*}
        c(\kappa)=\int_\kappa^{+\infty}\bar{F}(\zeta)\mathrm{d}\zeta=\int_\kappa^{+\infty}e^{-\left(-\ln\bar{F}(\zeta)\right)}\mathrm{d}\zeta\,\,,
    \end{equation*}
    thus, thanks to Theorem \ref{thm:Bingham}, with $g(\bullet)=-\ln\bar{F}(\bullet)$
    \begin{equation}
    \label{eq:as_logbarF}
        -\ln c(\kappa)=-\ln\int_\kappa^{+\infty}e^{-\left(-\ln\bar{F}(\zeta)\right)}\mathrm{d}\zeta\sim-\ln\bar{F}(\kappa)\in\mathcal{R}_\theta\,\,.
    \end{equation}

    \smallskip

    Let us prove that ii) $\Rightarrow$ iii). We first recall the following bounds for the function $\Phi(-x)$, for $x>0$ \citep[see, e.g.,][Sec.2.26, p.179]{mitrinovic1970analytic}
    \begin{equation*}
        \varphi(-x)\frac{1}{x+\sqrt{x^2+2}}\leq\Phi(-x)\leq\varphi(-x)\frac{1}{x+\sqrt{x^2+\frac{4}{\pi}}}\,\,.
    \end{equation*}
    From these inequalities, we obtain the following bounds for the normalised Bachelier formula \eqref{eq:norm_Bachelier} for any positive $\kappa$
    \begin{equation*}
        \varphi\left(-d\right)\left\{\sigma-\frac{\kappa}{d+\sqrt{d^2+\frac{4}{\pi}}}\right\}\leq c_b(\kappa,\sigma)\leq\varphi\left(-d\right)\left\{\sigma-\frac{\kappa}{d+\sqrt{d^2+2}}\right\}\,\,,
    \end{equation*}
    with $\displaystyle d:=\frac{\kappa}{\sigma}$. Thus, considering the implied volatility $\mathrm{I}(\kappa)$ in \eqref{eq:IV_eq} and $\displaystyle d(\kappa):=\frac{\kappa}{\mathrm{I}(\kappa)}$, we obtain
    \begin{equation}
\label{eq:inequality_mills}
        \varphi\left(-d(\kappa)\right)\left\{\mathrm{I}(\kappa)-\frac{\kappa}{d(\kappa)+\sqrt{d(\kappa)^2+\frac{4}{\pi}}}\right\}\leq c(\kappa)\leq\varphi\left(-d(\kappa)\right)\left\{\mathrm{I}(\kappa)-\frac{\kappa}{d(\kappa)+\sqrt{d(\kappa)^2+2}}\right\}\,\,.
    \end{equation}
    Let us observe that
    \begin{equation}
\label{eq:arg_log}
    \begin{cases}\displaystyle
        \mathrm{I}(\kappa)-\frac{\kappa}{d(\kappa)+\sqrt{d(\kappa)^2+\frac{4}{\pi}}}&=\displaystyle\frac{\mathrm{I}(\kappa)\sqrt{\kappa^2+\frac{4}{\pi}\mathrm{I}(\kappa)^2}}{\kappa+\sqrt{\kappa^2+\frac{4}{\pi}\mathrm{I}(\kappa)^2}}\vspace{3mm}
\\
\vspace{3mm}
\displaystyle
        \mathrm{I}(\kappa)-\frac{\kappa}{d(\kappa)+\sqrt{d(\kappa)^2+2}}&=\displaystyle\frac{\mathrm{I}(\kappa)\sqrt{\kappa^2+2\,\mathrm{I}(\kappa)^2}}{\kappa+\sqrt{\kappa^2+2\,\mathrm{I}(\kappa)^2}}
    \end{cases}\;\;.
    \end{equation}

    If we substitute \eqref{eq:arg_log} in \eqref{eq:inequality_mills} and apply the natural logarithm, we get
    \begin{align*}
        -\frac{\kappa^2}{2\,\mathrm{I}(\kappa)^2}-\ln\left(\sqrt{2\pi}\right)&+\ln\left(\frac{\mathrm{I}(\kappa)\sqrt{\kappa^2+\frac{4}{\pi}\mathrm{I}(\kappa)^2}}{\kappa+\sqrt{\kappa^2+\frac{4}{\pi}\mathrm{I}(\kappa)^2}}\right)\leq\ln c(\kappa)\leq\\&\leq\displaystyle-\frac{\kappa^2}{2\,\mathrm{I}(\kappa)^2}-\ln\left(\sqrt{2\pi}\right)+\ln\left(\frac{\mathrm{I}(\kappa)\sqrt{\kappa^2+2\,\mathrm{I}(\kappa)^2}}{\kappa+\sqrt{\kappa^2+2\,\mathrm{I}(\kappa)^2}}\right)
    \end{align*}
    \begin{equation}
    \label{eq:bounds_rewritten}
        \Leftrightarrow\,\ln\left(\frac{\mathrm{I}(\kappa)\sqrt{\kappa^2+\frac{4}{\pi}\mathrm{I}(\kappa)^2}}{\kappa+\sqrt{\kappa^2+\frac{4}{\pi}\mathrm{I}(\kappa)^2}}\right)\leq\ln c(\kappa)+\frac{\kappa^2}{2\,\mathrm{I}(\kappa)^2}+\ln{\sqrt{2\pi}}\leq\displaystyle\ln\left(\frac{\mathrm{I}(\kappa)\sqrt{\kappa^2+2\,\mathrm{I}(\kappa)^2}}{\kappa+\sqrt{\kappa^2+2\,\mathrm{I}(\kappa)^2}}\right)\,\,.
    \end{equation}
    We define
    \begin{equation}
    \label{eq:epsilon}
    \begin{cases}
\displaystyle        \varepsilon_1(\kappa):=&\ln c(\kappa)+\displaystyle\frac{\kappa^2}{2\,\mathrm{I}(\kappa)^2}\\[3mm]
\displaystyle        \varepsilon_2(\kappa):=&\varepsilon_1(\kappa)+\ln{\sqrt{2\pi}}
    \end{cases}\,\,.
    \end{equation}
    If $\varepsilon_1(\kappa)=o\left(\displaystyle\frac{\kappa^2}{2\,\mathrm{I}(\kappa)^2}\right)$ as $\kappa\to+\infty$, then
    \begin{equation*}
        \ln c(\kappa)=-\frac{\kappa^2}{2\,\mathrm{I}(\kappa)^2}+o\left(\frac{\kappa^2}{2\,\mathrm{I}(\kappa)^2}\right)\quad\Leftrightarrow\quad\frac{\ln c(\kappa)}{\kappa}=-\frac{\kappa}{2\,\mathrm{I}(\kappa)^2}+o\left(\frac{\kappa}{2\,\mathrm{I}(\kappa)^2}\right)\,\,,
    \end{equation*}
    this proves that, as $\kappa\to+\infty$,
    \begin{equation}
    \label{eq:asymptotic_lnc_I2}
        \frac{\ln c(\kappa)}{\kappa}\sim-\frac{\kappa}{2\,\mathrm{I}(\kappa)^2}\quad\Leftrightarrow\quad\frac{\mathrm{I}(\kappa)^2}{\kappa}\sim-\frac{\kappa}{2\,\ln c(\kappa)}\sim-\frac{\kappa}{2\,\ln \bar{F}(\kappa)}\,\,,
    \end{equation}
    where the last asymptotic holds thanks to \eqref{eq:as_logbarF}.
    
\bigskip

    To conclude the proof, we have to demonstrate that $\varepsilon_1(\kappa)=o\left(\frac{\kappa^2}{2\,\mathrm{I}(\kappa)^2}\right)$ as $\kappa\to+\infty$. From \eqref{eq:bounds_rewritten}, we have
    \begin{equation}
    \label{eq:inequality_eps2}
        \ln\left(\mathrm{I}(\kappa)\right)+\ln\left(\frac{\sqrt{\kappa^2+\frac{4}{\pi}\mathrm{I}(\kappa)^2}}{\kappa+\sqrt{\kappa^2+\frac{4}{\pi}\mathrm{I}(\kappa)^2}}\right)\leq\varepsilon_2(\kappa)\leq\displaystyle\ln\left(\mathrm{I}(\kappa)\right)+\ln\left(\frac{\sqrt{\kappa^2+2\,\mathrm{I}(\kappa)^2}}{\kappa+\sqrt{\kappa^2+2\,\mathrm{I}(\kappa)^2}}\right)\,\,\textcolor{black}{.}
    \end{equation}
    We recall that if $-\ln c(\kappa) \in \mathcal{R}_\theta$ with $\theta \ge 1$, then the call price $c(\kappa)$ decays at least exponentially as $\kappa \to +\infty$. 
    Therefore, assumption a) of Lemma \ref{lem:tec_lem} is satisfied.\\
    \textcolor{black}{
    The inequalities in \eqref{eq:inequality_eps2} and Lemma \ref{lem:tec_lem} imply that
    \begin{equation}
    \label{eq:bounds_eps2}\lim_{\kappa\to+\infty}\varepsilon_2(\kappa)\displaystyle\frac{2\,\mathrm{I}(\kappa)^2}{\kappa^2}=\lim_{\kappa\to+\infty}\left(\ln(\mathrm{I}(\kappa))+\ln\left(\frac{1}{2}\right)\right)\displaystyle\frac{2\,\mathrm{I}(\kappa)^2}{\kappa^2}=0\,\,,
    \end{equation}
    where the first equality is a consequence of Lemma \ref{lem:tec_lem}i) and the second of Lemma \ref{lem:tec_lem}ii).\\
    This means that
    }
    \begin{equation*}
       	\varepsilon_2(\kappa)=o\left(\frac{\kappa^2}{2\,\mathrm{I}(\kappa)^2}\right)\Rightarrow\varepsilon_1(\kappa)=o\left(\frac{\kappa^2}{2\,\mathrm{I}(\kappa)^2}\right)\quad\text{as}\;\,\kappa\to+\infty
    \end{equation*}
where the last implication is due to the definition of $\varepsilon_1(\kappa)$ and $\varepsilon_2(\kappa)$ in \eqref{eq:epsilon}
\end{proof}
In the Theorem above, we have shown the relation between the implied volatility for large positive moneyness degree and the returns' distribution. Now, in the same way, we can establish a connection between the implied volatility for large negative moneyness degrees and returns' distribution that extends Theorem 1.2 in \cite{benaim2009regular} to the Bachelier framework.
\begin{theorem}\label{prop:left_wing}{\rm (Left wing formula)}\\
    Assume $\theta\geq1$. Then, for $\kappa\to+\infty$, the following implications hold
    \begin{equation*}
        \text{i)}\;\;-\ln F(-\kappa)\in\mathcal{R}_\theta\,\,\Rightarrow\,\,\text{ii)}\;\;
\begin{cases}
-\ln p(-\kappa)\sim-\ln F(-\kappa)&\\
-\ln p(-\kappa)\in\mathcal{R}_\theta&
\end{cases}
\hspace{-0.5cm}
\,\,\Rightarrow\,\,\text{iii)}\;\;\frac{\mathrm{I}(-\kappa)^2}{\kappa}\sim -\frac{\kappa}{2\ln F(-\kappa)}\,\,.
    \end{equation*}
\end{theorem}
\begin{proof}
    The proof is analogous to Theorem \ref{prop:right_wing}, using Lemmas \ref{lem:bcl_price_barF} and \ref{lem:tec_lem_left}
\end{proof}
As a consequence of Theorems \ref{prop:right_wing} and \ref{prop:left_wing}, it holds that if $-\ln\bar{F}(\kappa)$ and $-\ln F(-\kappa)\in\mathcal{R}_\theta$, then
\begin{equation*}
    \displaystyle\frac{\mathrm{I}(\kappa)^2}{\kappa}\sim\begin{cases}&\displaystyle-\frac{\kappa}{2\,\ln \bar{F}(\kappa)}\quad\text{as}\;\,\kappa\to+\infty\,\,,\\[5mm]
    &\displaystyle-\frac{\kappa}{2\,\ln F(\kappa)}\quad\text{as}\;\,\kappa\to-\infty\end{cases}\,\,.
\end{equation*}
In many models of practical interest, the asymptotic behaviour of the returns' distribution is not explicitly known, while the characteristic function  of the random variable $\zeta$, defined as
\[
\phi(\xi):=\mathbb{E}\left[e^{i\xi\zeta}\right]=\int_{-\infty}^{+\infty}e^{i\xi\zeta}\mathrm{d}F(\zeta)\,\,,
\]
with $\xi$ a complex number with real part $\Re(\xi)$ and imaginary part $\Im(\xi)$, in most cases is explicitly known. We can establish a direct link between the implied volatility formula and the singularity structure of the characteristic function.

\smallskip

Theorem 3.1, p.12 in \cite{Lukacs} ensures that if the characteristic function $\phi(\xi)$ ``is regular in a neighborhood of the origin then it is also regular in a horizontal strip" bounded by two singularities located on the imaginary axis. We focus on the behaviour of $\phi(\xi)$ near the boundaries of its domain of analyticity. To this end, we introduce the following conditions.

\smallskip

\textbf{Condition I:} $\phi(\xi)$  is analytic in the horizontal strip $\Im(\xi)\in(-\lln,\lrn)$ with $\lambda_\pm > 0$ and finite.

\smallskip

\textbf{Condition II:} Condition I holds and, for some $n \geq 0$,
\textcolor{black}{
\[
(-i)^n\phi^{(n)}\left(\mp i\left(\lambda_\pm - s\right)\right) \sim s^{-\rho_\pm} \, l_\pm\left(\frac{1}{s}\right)\quad\text{as $s \to 0^+$}\;\;,
\]}
for some $\rho_\pm > 0$ and $l_\pm \in \mathcal{R}_0$, \textcolor{black}{where} $\phi^{(n)}(u)$ denotes the $n$-th derivative of the function $\phi(\bullet)$, evaluated at $u$.

\smallskip

The latter condition ensures that some derivatives of the characteristic function diverge at $\mp i\lambda_\pm$ in a regularly varying manner. As we will see later, the behaviour of the characteristic function in a neighbourhood of these two singularities is crucial, as summarised in Condition II. We state our result in the following Theorem.
\begin{theorem}\label{prop:analyticity_smile_2}
Consider $\lambda_\pm > 0$ finite. Then, the following implications hold:
\begin{equation*}
\text{i)}\quad
\text{Condition II}
\;\Rightarrow\;
\text{ii)}\quad
\begin{cases}
\ln \bar{F}(\kappa) \sim -\lambda_+ \kappa, \\
\ln F(-\kappa) \sim -\lambda_- \kappa
\end{cases}
\;\; \text{as } \kappa \to +\infty
\;\Rightarrow\;
\text{iii)}\quad
\lim_{\kappa \to \pm \infty} \frac{\mathrm{I}(\kappa)^2}{|\kappa|} = \frac{1}{2 \lambda_\pm}.
\end{equation*}
\end{theorem}

\begin{proof}
We prove that i) $\Rightarrow$ ii). Let us consider the moment generating function of the random variable $\zeta$:
\begin{equation*}
M(s) := \mathbb{E} \left[ e^{s\zeta} \right] = \phi(-i s), \quad s \in (-\lrn,\lln)\,\,.
\end{equation*}
Condition \textcolor{black}{II} is equivalent to
\[
M^{(n)}\left(\pm \left(\lambda_\pm - s\right)\right) \sim s^{-\rho_\pm} \, l_\pm\left(\frac{1}{s}\right)\quad\text{as $s \to 0^+$}\;\;,
\]
for some $\rho_\pm > 0$ and $l_\pm \in \mathcal{R}_0$, where $M^{(n)}$ denotes the $n$-th derivative of the moment generating function. Indeed
\[
M^{(n)}\left(\pm \left(\lambda_\pm - s\right)\right) = (-i)^n\phi^{(n)}\left(\mp i\left(\lambda_\pm - s\right)\right)\;\;.
\]

\smallskip

Thus, we can apply Criterion $1$ from \cite{benaim2008smile}, and conclude that
\[
\begin{cases}
\ln \bar{F}(\kappa) \sim -\lambda_+ \kappa, \\
\ln F(-\kappa) \sim -\lambda_- \kappa
\end{cases}\quad \text{as}\;\;\kappa\to+\infty\;\;.
\]

\smallskip

Finally, we prove that ii) $\Rightarrow$ iii).
Thanks to ii), $-\ln\bar{F}(\kappa)$ and $-\ln F(-\kappa)\in\mathcal{R}_1$, and we can apply Theorems \ref{prop:right_wing} and \ref{prop:left_wing}
\begin{align*}
    \lim_{\kappa\to+\infty}\frac{\mathrm{I}(\kappa)^2}{\kappa}&=-\lim_{\kappa\to+\infty}\frac{\kappa}{2\,\ln\bar{F}(\kappa)}=\frac{1}{2\,\lambda_+}\,\,,\\
    \lim_{\kappa\to-\infty}\frac{\mathrm{I}(\kappa)^2}{\kappa}&=-\lim_{\kappa\to-\infty}\frac{\kappa}{2\,\ln{F}(\kappa)}=-\frac{1}{2\,\lambda_-}
\end{align*}
\end{proof}

Condition II is fairly general and is satisfied by most models commonly used to describe returns in quantitative finance, in which the characteristic function is analytic on the horizontal strip $\Im(\xi)\in(-\lambda_+,\lambda_-)$. This includes stochastic volatility models \citep[e.g.,][]{heston1993closed}, L\'evy models such as Meixner \citep[see, e.g.,][]{grigelionis1999processes}, Tempered Stable, Normal Tempered Stable and Generalized Hyperbolic models \citep[see, e.g.,][Ch.4]{Cont}, and additive models \citep[e.g.,][]{carr2007self,ATS,baviera2026additive}. \textcolor{black}{The statement for the Heston model and NTS L\'evy models} has been proven by \cite{benaim2008smile}, while the other cases are discussed in Section \ref{sec:examples}. \\ 
For models that generalise the Bachelier framework,\footnote{Theorem \ref{prop:analyticity_smile_2} is even more general. Its first implication can be applied to models --commonly used in quantitative finance-- defined on log-returns, such as the Heston model, exponential L\'evy models \citep[cf.,][Ch.8]{Cont} or exponential additive models \citep[cf.,][]{ATS}.} 
Theorem \ref{prop:analyticity_smile_2} implies that, when at least one derivative of the characteristic function exhibits regularly varying divergence at $\mp i\lambda_{\pm}$, the implied variance grows linearly for large absolute values of moneyness degree $\kappa$, with a slope inversely proportional to the boundaries of the analyticity strip of the characteristic function, as observed in all derivatives markets.

\bigskip

We now prove that the converse also holds. Suppose that the implied variance grows linearly for large absolute values of the moneyness degree $\kappa$. Then the probability distribution has exponentially decaying tails and the characteristic function of returns admits an analyticity strip (Condition I). This is a more general condition than Condition II, and the boundaries of the strip are inversely proportional to the slope of the implied variance.
\begin{theorem}\label{thm:analyticity_smile_converse}
Consider $\lambda_\pm$ two positive and finite constants. Then, the following implications hold:
\begin{equation*}
\text{i)}\quad\lim_{\kappa \to \pm \infty} \frac{\mathrm{I}(\kappa)^2}{|\kappa|} = \frac{1}{2 \lambda_\pm}
\;\Rightarrow\;
\text{ii)}\quad
\begin{cases}
\ln \bar{F}(\kappa) \sim -\lambda_+ \kappa, \\
\ln F(-\kappa) \sim -\lambda_- \kappa
\end{cases}
\;\; \text{as } \kappa \to +\infty
\;\Rightarrow\;
\text{iii)}\;\;\begin{aligned}
    &\text{Condition I}.
\end{aligned}
\end{equation*}
\end{theorem}
\begin{proof}
    First, we prove that i)$\,\Rightarrow\,$ii). From the right side of \eqref{eq:asymptotic_lnc_I2}, that holds because in this case assumption b) of Lemma \ref{lem:tec_lem} is true,
    \begin{equation}\label{eq:log_c_asym}
        -\ln c(\kappa)\sim\frac{\kappa^2}{2\,\mathrm{I}(\kappa)^2}\sim\lln\kappa\quad\text{as }\kappa\to+\infty\,\,,
    \end{equation}
    thus,
    \begin{equation*}
        -\ln c(\kappa)\in\mathcal{R}_1\,\,.
    \end{equation*}
    Let us observe, from Lemma \ref{lem:bcl_price_barF}, that
    \begin{equation*}
        c'(\kappa)=-\Bar{F}(\kappa)\,\,,
    \end{equation*}
    and that
    \begin{equation*}
        (-\ln c(\kappa))'=\frac{\Bar{F}(k)}{c(k)}\,\,\quad\Rightarrow\quad-\ln(\Bar{F}(\kappa))=-\ln\left((-\ln c(\kappa))'\right)-\ln(c(\kappa))\,\,.
    \end{equation*}
    Since $-\ln c(\kappa)\sim\lln\kappa$ as $\kappa\to+\infty$ for \eqref{eq:log_c_asym} and $(-\ln c(\kappa))'$ is a positive function of $\kappa$, we can conclude that
    \begin{equation*}
        -\ln\left((-\ln c(\kappa))'\right)=o(\kappa)\quad\text{as }\kappa\to+\infty\quad\Rightarrow\quad-\ln\Bar{F}(\kappa)\sim\lln\kappa\quad\text{as }\kappa\to+\infty\,\,,
    \end{equation*}
    indeed, applying De l'H\^opital's rule (whose hypotheses are satisfied thanks to the monotonicity of $\Bar{F}(\kappa)$ and $c(\kappa)$,
    \[
    \lln=\lim_{\kappa\to+\infty}\frac{-\ln c(\kappa)}{\kappa}=\lim_{\kappa\to+\infty}\frac{(-\ln c(\kappa))'}{1}=\lim_{\kappa\to+\infty}\frac{\Bar{F}(\kappa)}{c(k)}\quad\Rightarrow\quad\lim_{\kappa\to+\infty}\ln(-\ln\Bar{F}(\kappa))')=\ln(\lln)\,\,.
    \]
    Analogously, we can prove that
    \begin{equation*}
        -\ln F(-\kappa)\sim\lrn\kappa\quad\text{as }\kappa\to+\infty\,\,.
    \end{equation*}
    Let us prove that ii)$\,\Rightarrow\,$iii). From condition ii), we have that 
    \[
    \lim_{\kappa\to+\infty} -\frac{\ln \bar{F}(\kappa)}{\kappa} = \lambda_+\,\,.
    \]
    By definition of the limit, for any arbitrary $a$ satisfying $0 < a < \lambda_+$, there exists $\kappa_a$ such that for all $\kappa > \kappa_a$, the inequality $\bar{F}(\kappa) \le e^{-a \kappa}$ holds. Indeed, choosing $\epsilon > 0$ such that $a = \lambda_+ - \epsilon$, for sufficiently large $\kappa$ we have
    \[
    -\ln \bar{F}(\kappa) \ge (\lambda_+ - \epsilon)\,\kappa\,\,,
    \]
    which implies 
    \[
    \bar{F}(\kappa) \le e^{-(\lambda_+ - \epsilon)\kappa} = e^{-a \kappa}\,\,.
    \]
    Now consider the characteristic function $\phi(\xi)$ for a complex argument $\xi = u + i v$. The existence of the expectation depends on the convergence of the integral involving the real part of the exponent, i.e., 
    \begin{equation}
    \label{eq:cf_integral}
        -\int_{0}^{+\infty} e^{-v \zeta} \mathrm{d}\Bar{F}(\zeta)\,\,.
    \end{equation}
    Focusing on the positive semi-axis, the integral \eqref{eq:cf_integral} converges absolutely if $\bar{F}(\kappa)$ decays faster than $e^{v\kappa}$. Using integration by parts and the bound derived above, we obtain:
    \[
    -\int_0^{+\infty} e^{-v \zeta} \mathrm{d}\Bar{F}(\zeta) = \left[e^{-v\zeta}\bar{F}(\zeta)\right]_0^{+\infty} + v \int_0^{+\infty} e^{-v\zeta}\bar{F}(\zeta)\mathrm{d}\zeta < +\infty
    \]
    provided that $v > -\lambda_+$, since we can always choose $a$ such that $-v < a < \lambda_+$.
    
    \smallskip
    
    Analogous reasoning applies to the left tail $F(-\kappa)$, ensuring convergence for $v < \lambda_-$. Consequently, $\phi(\xi)$ is analytic in the horizontal strip $\Im(\xi) \in (-\lambda_+, \lambda_-)$.
\end{proof}

The result above is crucial from both a theoretical and a practical perspective. Theorem \ref{thm:analyticity_smile_converse} shows that, if the Bachelier implied volatility $\mathrm{I}(\kappa)$ is proportional to the square root of the moneyness degree for large absolute values of $\kappa$, then the model for the underlying returns must belong to the class of models whose characteristic function is analytic in a horizontal strip of the complex plane. This has direct implications for model selection: one may first test whether the implied volatility, which is directly observable in the market, exhibits the behaviour described in point i) of Theorem \ref{thm:analyticity_smile_converse}, and then select the model accordingly.

\smallskip

In the remaining part of this Section, we show some examples of models for which Condition II 
holds.

\subsection{Examples}\label{sec:examples}
In this Section, we show some of the L\'evy and additive models already cited above that satisfy Condition II. In particular, we consider: the Meixner process, the Tempered Stable processes and the Generalized Hyperbolic processes.\\
\textcolor{black}{Consistently with the modelling framework adopted throughout the paper, throughout this Section the underlying $F_t$ is modelled arithmetically, as in \eqref{eq:fwd_dynamics}.
}
\subsubsection{Meixner}
The characteristic function of a Meixner process $\{F_t\}_{t\geq0}$ \citep[extending for a generic $F_0$,][Def.1, p.33]{grigelionis1999processes} is 
\begin{equation*}
    \phi_F(\xi):=\left(\frac{\cos\frac{\beta}{2}}{\cosh\frac{\alpha\xi-i\beta}{2}}\right)^{2\delta t}e^{i\xi\left(\mu t+F_0\right)}\,\,,
\end{equation*}
where $\alpha>0$, $-\pi<\beta<\pi$, $\mu\in\mathbb{R}$ and $\delta>0$. The characteristic function of returns $\zeta$ is
\begin{equation*}
    \phi(\xi)=\phi_F\left(\frac{\xi}{\sqrt{t}}\right)=\left(\frac{\cos\frac{\beta}{2}}{\cosh\frac{\alpha\xi-i\beta\sqrt{t}}{2\sqrt{t}}}\right)^{2\delta t}e^{i\xi\left(\mu\sqrt{t}+\frac{F_0}{\sqrt{t}}\right)}\,\,.
\end{equation*}
The characteristic function of $\zeta$ is analytic in the horizontal strip $\Im(\xi)\in(-\lln,\lrn)$ with
\[
\lambda_\pm=\frac{\pi\mp\beta}{\alpha}\sqrt{t}\,\,.
\]
With straightforward computations, one can prove that
\begin{equation*}
    (-i)\phi'(\mp i(\lambda_\pm-s))\sim\sqrt{t}\,\left(\frac{\alpha\,s}{2\sqrt{t}}\right)^{-1-2\delta t}\cos\frac{\beta}{2}\,e^{\pm\lambda_\pm\left(\mu\sqrt{t}+\frac{F_0}{\sqrt{t}}\right)}\quad\text{as }s\to0^+\,\,,
\end{equation*}
thus, Condition II is satisfied with
\begin{align*}
    &\rho_\pm=1+2\delta t>0\\
    &l_\pm\left(\frac{1}{s}\right)=\sqrt{t}\,\left(\frac{\alpha}{2\sqrt{t}}\right)^{-1-2\delta t}\cos\frac{\beta}{2}\,e^{\pm\lambda_\pm\left(\mu\sqrt{t}+\frac{F_0}{\sqrt{t}}\right)}\in\mathcal{R}_0\,\,,
\end{align*}
and $n=1$. Furthermore, this result naturally extends to the class of additive processes. Indeed, at any fixed maturity $t$, the time-inhomogeneity of the process does not alter the singularity structure of the characteristic function, ensuring that Condition II remains satisfied.
\subsubsection{Tempered Stable}
The characteristic function of a Tempered Stable process $\{F_t\}_{t \ge 0}$ \citep[see, e.g.,][Prop. 4.2]{Cont} is
\begin{equation*}
    \phi_F(\xi) := \exp\left( t\,C\, \Gamma(-\alpha) \left[ (a_+ - i\xi)^\alpha - a_+^\alpha + (a_- + i\xi)^\alpha - a_-^\alpha \right] + i\xi\left(\mu t+F_0\right) \right)\,\,,
\end{equation*}
where $\alpha \in (0, 2) \setminus \{1\}$, $C > 0$, and $a_+, a_- > 0$. The characteristic function of returns $\zeta$ is
\begin{equation*}
    \phi(\xi) = \phi_F\left(\frac{\xi}{\sqrt{t}}\right) = \exp\left( t\,C\, \Gamma(-\alpha) \left[ \left(a_+ - i\frac{\xi}{\sqrt{t}}\right)^\alpha - a_+^\alpha + \left(a_- + i\frac{\xi}{\sqrt{t}}\right)^\alpha - a_-^\alpha \right] + i\xi\left(\mu \sqrt{t}+\frac{F_0}{\sqrt{t}}\right) \right)\,\,.
\end{equation*}
The characteristic function of $\zeta$ is analytic in the horizontal strip $\Im(\xi) \in (-\lambda_+, \lambda_-)$ with
\[
\lambda_{\pm} = a_{\pm}\sqrt{t} \;.
\]
With straightforward computations, considering the case $0 < \alpha < 1$, one can prove that
\begin{equation*}
    (-i)\phi'(\mp i(\lambda_{\pm} - s)) \sim s^{\alpha-1} \sqrt{t}\,C\,\Gamma(1-\alpha)\,t^{-\alpha/2} \phi(\mp i\lambda_{\pm}) \quad \text{as } s \to 0^+ \;,
\end{equation*}
thus, Condition II is satisfied with
\begin{align*}
    &\rho_{\pm} = 1 - \alpha > 0\\
    &l_{\pm}\left(\frac{1}{s}\right) =\displaystyle \sqrt{t}\,C\,\Gamma(1-\alpha)\,t^{-\alpha/2}\,e^{t C \Gamma(-\alpha)[(a_{\mp} + a_{\pm})^\alpha - a_{\mp}^\alpha - a_{\pm}^\alpha] \pm \lambda_{\pm}\left(\mu \sqrt{t}+\frac{F_0}{\sqrt{t}}\right)} \in \mathcal{R}_0 \;,
\end{align*}
and $n=1$. In the case $1 < \alpha < 2$, the first derivative converges at the boundaries. However, the second derivative diverges. Straightforward computations show that:
\begin{equation*}
    (-i)^2 \phi''(\mp i(\lambda_{\pm} - s)) \sim s^{\alpha-2} \sqrt{t}\,C\, \Gamma(2-\alpha)\,t^{-(1+\alpha)/2} \,\phi(\mp i\lambda_{\pm}) \quad \text{as } s \to 0^+ \;,
\end{equation*}
thus, Condition II is satisfied with $n=2$ and:
\begin{align*}
    &\rho_{\pm} = 2 - \alpha > 0\\
    &l_{\pm}\left(\frac{1}{s}\right) = \sqrt{t}\,C\,\Gamma(2-\alpha)\,t^{-(1+\alpha)/2}\,\phi(\mp i\lambda_{\pm}) \in \mathcal{R}_0 \;.
\end{align*}
As observed in the Meixner case, this Condition applies to additive processes as well.
\subsubsection{Generalized Hyperbolic models}
The characteristic function of a Generalized Hyperbolic process $\{F_t\}_{t \ge 0}$ \citep[see, e.g.,][eq.(4.39), Sec.4.6]{Cont} is given by
\begin{equation*}
\phi_F(\xi) := e^{i\xi\left(\mu t+F_0\right)} \left( \frac{\alpha^2 - \beta^2}{\alpha^2 - (\beta + i\xi)^2} \right)^{\frac{\lambda}{2}} \frac{K_{\lambda}(\delta \sqrt{\alpha^2 - (\beta + i\xi)^2})}{K_{\lambda}(\delta \sqrt{\alpha^2 - \beta^2})}\,\,,
\end{equation*}
where $K_{\lambda}(\bullet)$ denotes the modified Bessel function of the second kind \citep[see, e.g.,][Ch.9, p.374]{abramowitz1948handbook}, and the parameters satisfy $\alpha > |\beta|$, $\delta > 0$ and $\lambda \in \mathbb{R}$.\\
The characteristic function of returns $\zeta$ is $\phi(\xi) = \phi_F(\xi/\sqrt{t})$ and is analytic in the horizontal strip $\Im(\xi) \in (-\lambda_+, \lambda_-)$ with boundaries $\lambda_{\pm} = (\alpha \mp \beta)\sqrt{t}$. The behaviour near the boundary depends on $\lambda$.

\smallskip

The limiting case $\lambda \to 0$, corresponding to the Variance Gamma process, has been discussed in \cite{benaim2008smile}.

\smallskip

For $\lambda > 0$, the characteristic function diverges at the boundaries; specifically, using the asymptotic relation $K_{\lambda}(z) \sim \frac{1}{2}\Gamma(\lambda)(z/2)^{-\lambda}$ as $z \to 0$ \citep[cf.][eq.(9.6.9), Ch.9, p.375]{abramowitz1948handbook} and noting that the argument of the Bessel function scales as $\sqrt{s}$ near the boundary (where $s$ is the distance from the singularity), the function behaves as $\phi(-i(\lambda_+ - s)) \sim s^{-\lambda}$. Thus, Condition II is satisfied with $n=0$ and $\rho_{\pm} = \lambda$.

\smallskip

Conversely, for $\lambda < 0$, we define $\nu = -\lambda > 0$. The relevant term in the characteristic function takes the form $z^{\nu} K_{\nu}(z)$ with $z \sim \sqrt{s}$. Utilising the series expansion for the Bessel function of small argument \citep[cf.][eq.(9.6.11), Ch.9, p.375]{abramowitz1948handbook}, the leading term is constant (implying that the function is bounded), while the dominant singular term scales as $z^{2\nu} \sim s^{\nu} = s^{-\lambda}$. Consequently, the function itself converges, but its derivatives diverge. To satisfy Condition II, one must choose the differentiation order $n$ such that $n > -\lambda$. The $n$-th derivative scales as $s^{-\lambda - n}$, satisfying the condition with a tail decay index $\rho_{\pm} = n + \lambda$. For instance, in the Normal Inverse Gaussian case ($\lambda = -1/2$), we have $n=1$ and $\rho_{\pm} = 1/2$.\\
Similarly, for additive processes of this class, Condition II holds because the time inhomogeneity does not affect the characteristic function's singularities in the complex plane.

\textcolor{black}{
\subsection{Numerical results}
\label{sec:numerics}
The results of Section \ref{sec:smiles} are asymptotic in the moneyness degree $\kappa$, and it is therefore natural to investigate how closely they are attained at finite moneyness degrees, in line with the numerical analysis carried out in \cite{medvedev2007approximation,medvedev2010pricing}. We consider the NIG and the Meixner processes of Section \ref{sec:examples}, with three parameter sets each, and we study the behaviour of the implied variance for large absolute moneyness degree, at maturity $t=1$.}

\smallskip

\textcolor{black}{
By Theorem \ref{prop:analyticity_smile_2}, the implied variance of both models grows linearly in the moneyness degree,
\begin{equation*}
  \frac{I(\kappa)^2}{|\kappa|}\;\longrightarrow\;\frac{1}{2\lambda_\pm}\quad\text{as}\quad
  \kappa\to\pm\infty\,,
\end{equation*}
where $\lambda_\pm$ are the boundaries of the analyticity strip of the characteristic function.\\
In Figure~\ref{fig:numerics}, we plot $I(\kappa)^2/|\kappa|$ for the two models and the corresponding parameter sets; the horizontal lines mark the limits $1/(2\lambda_+)$ and $1/(2\lambda_-)$. In every case the implied variance approaches its limit smoothly and monotonically in both wings, providing a clear numerical confirmation of Theorem \ref{prop:analyticity_smile_2}.
}

\smallskip

\textcolor{black}{
\begin{figure}[H]
  \centering
  \includegraphics[width=\textwidth]{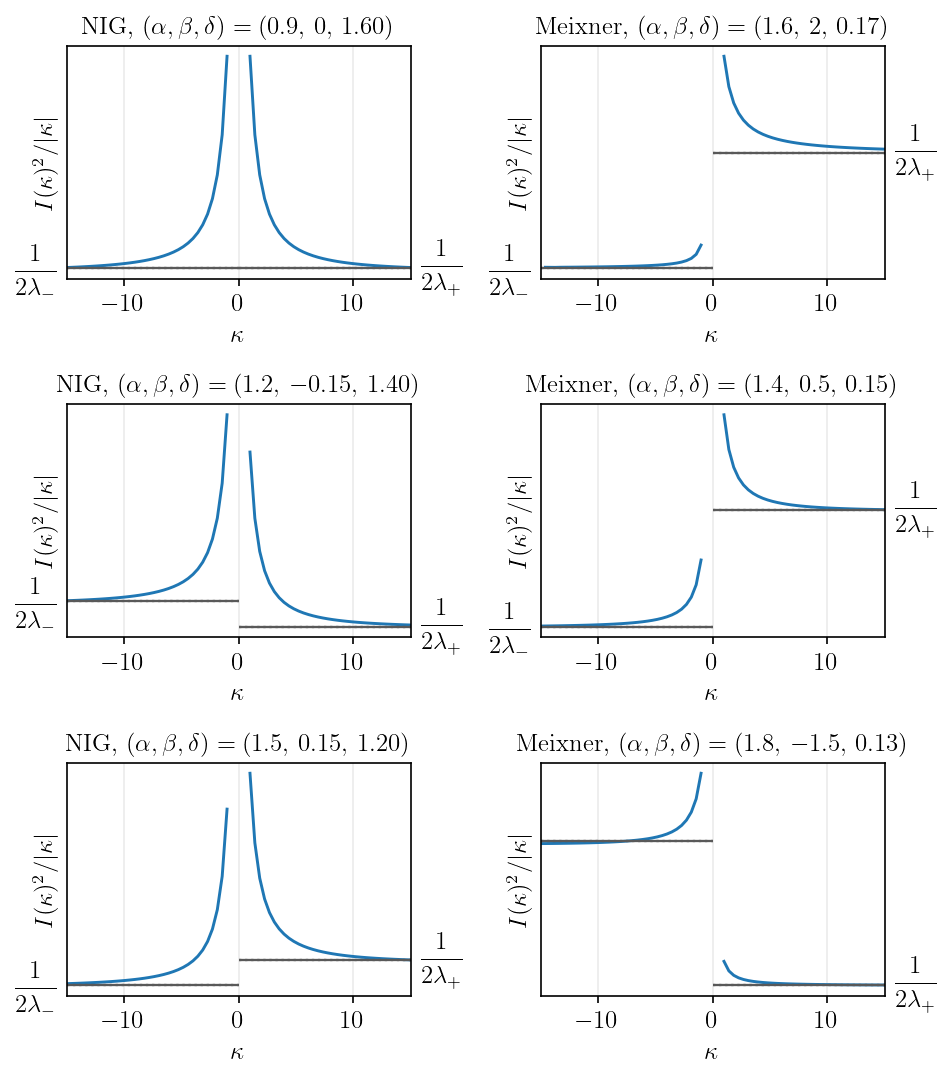}
  \caption{\color{black} Implied variance $I(\kappa)^2/|\kappa|$ as a function of the moneyness degree $\kappa$, at maturity $t=1$, for the NIG (left column) and the Meixner (right column) models. Each row corresponds to a different parameter set $(\alpha,\beta,\delta)$, reported in the panel title. In every panel the blue line is $I(\kappa)^2/|\kappa|$, while the two black lines mark the asymptotic limits $1/(2\lambda_+)$ (right wing, $\kappa\to+\infty$) and $1/(2\lambda_-)$ (left wing, $\kappa\to-\infty$) of Theorem \ref{prop:analyticity_smile_2}, where $\lambda_\pm$ are the boundaries of the model's analyticity strip. In both wings, and for all the parameter sets considered, the implied variance approaches its limit smoothly and monotonically, and is already close to it at moderately large $|\kappa|$.}
  \label{fig:numerics}
\end{figure}
The convergence is, moreover, fast: the curves are already close to their limiting value at moderately large $|\kappa|$, showing that the asymptotic slope of the implied variance is not merely a large-moneyness theoretical result, but also an accurate approximation at practically relevant moneyness degrees.
}
\section{Conclusions}
\label{sec:Conclusions}
In this paper, we have extended the asymptotic analysis of the implied volatility smile from the Black–Scholes implied volatility, as developed by \cite{benaim2009regular}, to the Bachelier implied volatility, which is particularly relevant in markets where the underlying asset can take negative values.

\smallskip

Using the theory of regular variation, we have derived explicit asymptotic expressions for the Bachelier implied volatility in the wings of the volatility smile, i.e., for large positive or negative moneyness degree. We have proven that there exists a direct link between the tail behaviour of the distribution of the underlying's returns and the corresponding implied volatility behaviour.

\smallskip

Moreover, in cases where the distribution of returns is not available in closed form but the characteristic function is known, we have shown that our formula can still characterise the asymptotic behaviour of the implied volatility. In particular, we have established that the boundaries of the analyticity strip of the returns' characteristic function directly determine the behaviour of Bachelier implied volatility in the smile's wings for most common models in quantitative finance.

\smallskip

Finally, we have shown that if the implied variance grows linearly for large absolute values of the moneyness degree, then the tails of the returns’ probability distribution necessarily exhibit exponential decay and the corresponding characteristic function is analytic in a horizontal strip of the complex plane. Moreover, both the decay rates and the boundaries of the strip are determined by the slope of the implied variance at large moneyness degree.

\section*{Acknowledgements}
We thank the Editor and the two anonymous Referees for their valuable comments and suggestions. We also thank all participants in the Advances in Mathematical Finance Conference in Freiburg for comments and discussions.

\bibliography{main}
\bibliographystyle{tandfx}
\section*{Notation and shorthands}
\begin{flushleft}

	\begin{tabular} {|c|l|}
		\toprule
		\textbf{Symbol}& \textbf{Description}\\ \bottomrule
		$c_b(y,\sigma), p_b(y,\sigma)$ & normalised Bachelier call (put) option price wrt moneyness $y$, with volatility $\sigma$ \\
        \textcolor{black}{$F_0,\,F_t$} & \textcolor{black}{value of the underlying at time $0$ and $t$}\\
        \textcolor{black}{$K$} & \textcolor{black}{option strike}\\
        $\kappa$ & option moneyness degree\\
        $t$ & time-to-maturity\\
		$c(\kappa), p(\kappa)$ & undiscounted call and put prices wrt moneyness degree $\kappa$\\
		$\mathrm{I}(\kappa)$ & implied volatility wrt the moneyness degree $\kappa$\\
        \textcolor{black}{$(\Omega,\mathcal{F},\mathbb{F},\mathbb{Q})$} & \textcolor{black}{filtered probability space} \\
        $\Re(\xi), \Im(\xi)$ & real and imaginary parts of a complex number $\xi$\\
		${\varphi} ( \bullet ), {\Phi} ( \bullet )$ & pdf and cdf of a standard normal rv\\
		$\zeta$ & increment of forward price divided by the square root of the time-to-maturity $t$\\
		$F(\bullet), \bar{F}(\bullet)$ & cdf of the rv $\zeta$ and $\bar{F}(\bullet)=1-F(\bullet)$\\
		$\phi(\bullet)$ & characteristic function of the rv $\zeta$\\
		$\phi^{(n)}(\bullet)$ & $n$-th derivative of the characteristic function $\phi(\bullet)$\\
        $\lambda_\pm$ & boundaries of the analyticity strip of $\phi(\bullet)$ and decay rates of distribution tails \\
		$M(s)$ & moment generating function of the rv $\zeta$\\
        $\theta$ & regular varying functions index \\
		$\mathcal{R}_\theta$ & class of regular varying functions with index $\theta>0$\\
        \textcolor{black}{$\overset{law}{=}$} & \textcolor{black}{equality in law} \\
		\bottomrule		
	\end{tabular}
	
\end{flushleft}	

\smallskip 

\begin{flushleft}
	\begin{tabular}{|c|l|}
		\toprule
		\textbf{Symbol}& \textbf{Description}\\ \bottomrule
		cdf & cumulative distribution function \\
		pdf & probability density function \\
		rv & random variable \\
            	wrt  & with respect to \\
		\bottomrule
	\end{tabular}
\end{flushleft}

\end{document}